\documentclass{article}
\usepackage{amssymb}
\usepackage{amsmath}
\usepackage{macros}
\usepackage{fullpage}
\usepackage{stmaryrd}
\usepackage{xspace}
\usepackage{macros}
\usepackage{algorithmic}
\usepackage{algorithm}

\newcommand{\temporaryremove}[1]{}
\def \BISI{\sim}
\def \CL{{\cal L}}
\newcommand\Real {{\mathbb R}}
\def \BP{\textbf{P}}
\def \CN{{\cal N}}
\def \IN{\it in}
\def \OUT{\it out}
\def \Var{{\it Var}}
\def \fv{{\it fv}}

\def \Pow#1{\mathcal P(#1)}
\newcommand{\Env}{{\sf Env}}
\def \CE{{\cal E}}
\newcommand{\boxm}[1]{[#1]}
\def \CM{{\cal M}}
\def\squareforqed{\hbox{\rlap{$\sqcap$}$\sqcup$}}
\def\qed{\ifmmode\squareforqed\else{\unskip\nobreak\hfil
\penalty50\hskip1em\null\nobreak\hfil\squareforqed
\parfillskip=0pt\finalhyphendemerits=0\endgraf}\fi}

\begin{document}
\title{Logical, Metric, and Algorithmic\\ Characterisations of Probabilistic Bisimulation}
\author{
Yuxin Deng$^{1}$ \qquad Wenjie Du$^2$\\
$^1$ Shanghai Jiao Tong University, China\\
$^2$ Shanghai Normal University, China}
\maketitle

\begin{abstract}
Many behavioural equivalences or preorders for probabilistic
processes involve a lifting operation that turns a relation on
states into a relation on distributions of states. We show that
several existing proposals for lifting relations can be reconciled
to be different presentations of essentially the same lifting
operation. More interestingly, this lifting operation nicely
corresponds to the Kantorovich metric, a fundamental concept used in
mathematics to lift a metric on states to a metric on distributions
of states, besides the fact the lifting operation is related to the
maximum flow problem in optimisation theory.

The lifting operation yields a neat notion of probabilistic
bisimulation, for which we provide logical, metric, and algorithmic
characterisations. Specifically, we extend the Hennessy-Milner logic
and the modal mu-calculus with a new modality, resulting in an
adequate and an expressive logic for probabilistic bisimilarity,
respectively. The correspondence of the lifting operation and the
Kantorovich metric leads to a natural characterisation of
bisimulations as pseudometrics which are post-fixed points of a
monotone function. We also present an ``on the fly" algorithm to
check if two states in a finitary system are related by
probabilistic bisimilarity, exploiting the close relationship
between the lifting operation and the maximum flow problem.
\end{abstract}

\section{Introduction}
In the last three decades a wealth of behavioural equivalences have
been proposed in concurrency theory. Among them, \emph{bisimilarity}
\cite{Mil89a,Par81} is probably the most studied one as it admits a
suitable semantics, an elegant co-inductive proof technique, as well
as efficient decision algorithms.

In recent years, probabilistic constructs have been proven useful
for giving quantitative specifications of system behaviour. The
first papers on probabilistic concurrency theory
\cite{GJS90,Chr90,larsenskou} proceed by \emph{replacing}
nondeterministic with probabilistic constructs. The reconciliation
of nondeterministic and probabilistic constructs starts with
\cite{HJ90} and has received a lot of attention in the literature
\cite{WL92,SL94,Lowe95,Seg95,HK97,MM97,BS01,JW02,mislove03,CIN05,TKP05,MM07,
DGHMZ07,DGMZ07,DGHM08,DGHM09}.

We shall also work in a framework that features the co-existence of
probability and nondeterminism. More specifically, we deal with
\emph{probabilistic labelled transition systems (pLTSs)}
\cite{DGHMZ07} which are an extension of the usual labelled
transition systems (LTSs) so that a step of transition is in the
form $s\ar{a}\Delta$, meaning that state $s$ can perform action $a$
and evolve into a distribution $\Delta$ over some successor states.
In this setting state $s$ is related to state $t$ by a relation
$\aRel$, say probabilistic simulation, written $s\aRel t$, if for
each transition $s\ar{a}\Delta$ from $s$ there exists a transition
$t\ar{a}\Theta$ from $t$ such that $\Theta$ can somehow simulate the
behaviour of $\Delta$ according to $\aRel$. To formalise the
mimicking of $\Delta$ by $\Theta$, we have to \emph{lift} $\aRel$ to
be a relation $\lift{\aRel}$ between distributions over states and
require $\Delta\lift{\aRel}\Theta$.

Various approaches of lifting relations have appeared in the
literature; see e.g. \cite{LS91,SL94,DGHMZ07,DD07,DGHM09}. We will
show that although those approaches appear different, they can be
reconciled. Essentially, there is only one lifting operation, which
has been presented in different forms. Moreover, we argue that the
lifting operation is interesting itself. This is justified by its
intrinsic connection with some fundamental concepts in mathematics,
notably \emph{the Kantorovich metric} \cite{Kan42}. For example, it
turns out that our lifting of binary relations from states to
distributions nicely corresponds to the lifting of metrics from
states to distributions by using the Kantorovich metric. In
addition, the lifting operation is closely related to \emph{the
maximum flow problem} in optimisation theory, as observed by Baier
et al. \cite{BEM00}.

A good scientific concept is often elegant, even seen from many
different perspectives. Bisimulation is one of such concepts in the
traditional concurrency theory, as it can be characterised in a
great many ways such as fixed point theory, modal logics, game
theory, coalgebras etc. We believe that probabilistic bisimulation
is also one of such concepts in probabilistic concurrency theory. As
an evidence, we will provide in this paper three characterisations,
from the perspectives of modal logics, metrics, and decision
algorithms.
\begin{enumerate}
\item Our logical characterisation of probabilistic bisimulation
consists of two aspects: \emph{adequacy} and \emph{expressivity}
\cite{Pnu85}. A logic $\CL$ is adequate when two states are
 bisimilar if and only if they satisfy exactly the same
set of formulae in $\CL$. The logic is expressive when each state
$s$ has a characteristic formula $\phi_s$ in $\CL$ such that $t$ is
bisimilar to $s$ if and only if $t$ satisfies $\phi_s$. We will
introduce a probabilistic choice modality  to capture the behaviour
of distributions. Intuitively, distribution $\Delta$ satisfies the
formula $\bigoplus_{i\in I}p_i\cdot\phi_i$ if there is a
decomposition of $\Delta$ into a convex combination some
distributions, $\Delta=\sum_{i\in I}p_i\cdot\Delta_i$, and each
$\Delta_i$ confirms to the property specified by $\phi_i$. When the
new modality is added to the Hennessy-Milner logic \cite{HM85} we
obtain an adequate logic for probabilistic bisimilarity; when it is
added to the modal mu-calculus \cite{Koz83} we obtain an expressive
logic.

\item By metric characterisation of probabilistic bisimulation, we
mean to give a pseudometric such that two states are bisimilar if
and only if their distance is $0$ when measured by the pseudometric.
More specifically, we show that bisimulations correspond to
pseudometrics which are post-fixed points of a monotone function,
and in particular bisimilarity corresponds to a pseudometric which
is the greatest fixed point of the monotone function.

\item As to the algorithmic characterisation, we propose an ``on the
fly" algorithm that checks if two states are related by
probabilistic bisimilarity. The schema of the algorithm is to
approximate probabilistic bisimilarity by iteratively accumulating
information about state pairs $(s,t)$ where $s$ and $t$ are not
bisimilar. In each iteration we dynamically constructs a relation
$\aRel$ as an approximant. Then we verify if every transition from
one state can be matched up by a transition from the other state,
and their resulting distributions are related by the lifted relation
$\lift{\aRel}$, which involves solving the maximum flow problem of
an appropriately constructed network, by taking advantage of the
close relation between our lifting operation and the above mentioned
maximum flow problem.
\end{enumerate}

\paragraph{Related work}
Probabilistic bisimulation was first introduced by Larsen and Skou
\cite{LS91}. Later on, it was investigated in a great many
probabilistic models.
An adequate logic for 
probabilistic bisimulation in a setting similar to our pLTSs has
been studied in \cite{JWL01,PS07}. It is also based on an probabilistic
extension of the Hennessy-Milner logic. The main difference from our
logic in Section~\ref{s:adequate} is the introduction of the
operator $[\Cdot ]_p$. Intuitively, a distribution $\Delta$
satisfies the formula $[\phi]_p$ when the set of states satisfying
$\phi$ is measured by $\Delta$ with probability at least $p$. So the
formula $[\phi]_p$ can be expressed by our logic in terms of the
probabilistic choice $\bigoplus_{i\in I}p_i\Cdot\phi_i$ by setting
$I=\{1,2\}$, $p_1=p$, $p_2=1-p$, $\phi_1=\phi$, and $\phi_2=true$.
When restricted to
deterministic pLTSs (i.e., for each state and for each action, there
exists at most one outgoing transition from the state),
probabilistic bisimulations can be characterised by simpler forms of
logics, as observed in \cite{LS91,DEP98,PS07}.

An expressive logic for nonprobabilistic 
bisimulation has been proposed in \cite{SI94}.
 In this paper we partially extend the results of
\cite{SI94} to a probabilistic setting that admits both
probabilistic and nondeterministic choice. We present a
probabilistic extension of the modal mu-calculus \cite{Koz83}, where
a formula is interpreted as the set of states satisfying it. This is
in contrast to the probabilistic semantics of the mu-calculus as
studied in
 \cite{HK97,MM97,MM07} where formulae denote lower bounds of
probabilistic evidence of properties, and the semantics of the
generalised probabilistic logic of \cite{CIN05} where a mu-calculus
formula is interpreted as a set of deterministic trees that satisfy
it.

The Kantorovich metric has been used by van Breugel \emph{et al.}
for defining behavioural pseudometrics on fully probabilistic
systems \cite{BW01,BW06,BSW07} and reactive probabilistic systems
\cite{BW01b,BW05,BHMW05,BHMW07}; and by Desharnais {\em et al.} for
labelled Markov chains \cite{DGJP99,DGJP04} and labelled concurrent
Markov chains \cite{DGJP02}; and later on by Ferns \emph{et al.} for
Markov decision processes \cite{FPP04,FPP05}; and by Deng \emph{et
al.} for action-labelled quantitative transition systems
\cite{DCPP05}.  One exception is \cite{DLT08}, which proposes a
pseudometric for labelled Markov chains without using the Kantorovich
metric.  Instead, it is based on a notition of $\epsilon$-bisimulation, which
relaxes the definition of probabilistic bisimulation by allowing small
perturbation of probabilities.
In this paper we are mainly interested in the
correspondence of our lifting operation to the Kantorovich metric.
The metric characterisation of probabilistic bisimulation in
Section~\ref{s:metric} is merely a direct consequence of this
correspondence.

Decision algorithms for probabilistic bisimilarity and similarity
have been considered by Baier et al. in \cite{BEM00} and Zhang et
al. in \cite{ZHEJ08}. Their algorithms are global in the sense that
a whole state space has to be fully generated in advance. In
contrast, ``on the fly" algorithms are local in the sense that the
state space is dynamically generated which is often more efficient
to determine that one state fails to be related to another. Our
algorithm in Section~\ref{s:algo} is inspired by \cite{BEM00}
because we also reduce the problem of checking if two distributions
are related by a lifted relation to the maximum flow problem of a
suitable network. We generalise the local algorithm of checking
nonprobabilistic bisimilarity \cite{FM90,Lin98} to the probabilistic
setting.

This paper provides a relatively comprehensive account of
probabilistic bisimulation. Some of the results or their variants were mentioned previously 
in \cite{DCPP05,DD09,DD09b,DG10}. Here they are presented in a uniform
way and equipped with detailed proofs.

\paragraph{Outline of the paper} The paper proceeds by recalling a
way of lifting binary relations from states to distributions, and
showing its coincidence with a few other ways in
Section~\ref{s:lift}. The lifting operation is justified in
Section~\ref{s:justify} in terms of its correspondence to the
Kantorovich metric and the maximum flow problem. In
Section~\ref{s:pbisi} we define probabilistic bisimulation and show
its infinite approximation. In Section~\ref{s:logic} we introduce a
probabilistic choice modality, then extend the Hennessy-Milner logic
and the modal mu-calculus so to obtain two logics that are adequate
and expressive, respectively. In Section~\ref{s:metric} we
characterise probabilistic bisimulations as pseudometrics. In
Section~\ref{s:algo} we exploit the correspondence of our lifting
operation to the maximum flow problem, and present a polynomial time
decision algorithm. Finally, Section~\ref{s:conclude} concludes the
paper.

\section{Lifting relations}\label{s:lift}
In the probabilistic setting, formal systems are usually modelled as
distributions over states. To compare two systems involves the
comparison of two distributions. So we need a way of lifting
relations on states to relations on distributions. This is used, for
example, to define probabilistic bisimulation as we shall see in
Section~\ref{s:pbisi}. A few approaches of lifting relations have
appeared in the literature. We will take the one from \cite{DGHM09},
and show its coincidence with two other approaches.

We first fix some notation. A (discrete) probability
distribution\index{probability distribution} over a set $S$ is a
mapping \mbox{$\Delta: S \rightarrow [0,1]$} with $\sum_{s\in
S}\Delta(s)=1$. The \emph{support}\index{support} of $\Delta$ is
given by $\support{\Delta} := \setof{s \in S}{\Delta(s) > 0}$. In
this paper we only consider finite state systems, so it suffices to
use distributions with finite support; let $\dist{S}$, ranged over
by $\Delta, \Theta$, denote the collection of all such distributions
over $S$.  We use $\pdist{s}$ to denote the point distribution,
satisfying $\pdist{s}(t)=1$ if $t=s$, and $0$ otherwise. If $p_i\geq
0$ and $\Delta_i$ is a distribution for each $i$ in some finite
index set $I$, then $\sum_{i \in I}p_i \ocdot \Delta_i$ is  given by
\begin{equation*}
  (\sum_{i \in I}p_i \ocdot \Delta_i)(s) ~~~=~~~ \sum_{i \in I} p_i \cdot\Delta_i(s)
\end{equation*}
If $\sum_{i \in I} p_i = 1$ then this is easily seen to be a
distribution in $\dist{S}$.
Finally, the
\emph{product}\index{product} of two probability distributions
$\Delta,\Theta$ over $S,T$ is the distribution $\Delta\times\Theta$
over $S\times T$ defined by
$(\Delta\times\Theta)(s,t):=\Delta(s)\cdot\Theta(t)$.

\begin{definition}\label{d:lifting}
Given two sets $S$ and $T$ and a relation $\mathord{\aRel} \subseteq
S \mathop{\times} T$. Then $\mathord{\lift{\aRel}} \subseteq
\dist{S} \mathop{\times} \dist{T}$ is the smallest relation that
satisfies:
\begin{enumerate}
\item
$s\aRel t$ implies $\pdist{s}\lift{\aRel}\pdist{t}$

\item
$\Delta_i\lift{\aRel}\Theta_i$ implies $(\sum_{i\in
I}p_i\cdot\Delta_i)\lift{\aRel} (\sum_{i\in I}p_i\cdot\Theta_i)$,
where $I$ is a finite index set and  $\sum_{i\in I}{p_i} = 1$.
\end{enumerate}
\end{definition}

The lifting construction satisfies the following  useful property
whose proof is straightforward thus omitted.
\begin{proposition}\rm\label{prop:lifting}
Suppose $\mathord{\aRel}\subseteq S\times S$ and
 $\sum_{i\in I} p_i = 1$. If $(\sum_{i\in I}{p_i \ocdot \Delta_i}) \lift{\aRel} \Theta$ then
$\Theta = \sum_{i\in I}{p_i \ocdot \Theta_i}$  for some set of
distributions $\Theta_i$ such that $\Delta_i \lift{\aRel} \Theta_i$.
\hfill\qed
\end{proposition}

We now look at alternative presentations of
Definition~\ref{d:lifting}. The proposition below is immediate.
\begin{proposition}\label{p:lifting}
Let $\Delta$ and $\Theta$ be distributions over $S$ and $T$,
  respectively, and $\aRel\subseteq S\times T$. Then
$\Delta \lift{\aRel} \Theta$ if and only if $\Delta,\Theta$ can be
decomposed as follows:
\begin{enumerate}
\item
$\Delta = \sum_{i\in I}{p_i \cdot \pdist{s_i}}$, where
  $I$ is a finite index set and  $\sum_{i\in I}{p_i} = 1$

\item
For each $i \in I$ there is a state $t_i$ such that $s_i \aRel t_i$

\item
$\Theta = \sum_{i \in I} {p_i \cdot \pdist{t_i}}$. \hfill\qed
\end{enumerate}
\end{proposition}

An important point here is that in the decomposition of $\Delta$
into $\sum_{i\in I}{p_i \ocdot \pdist{s_i}}$, the states $s_i$ are
\emph{not necessarily distinct}: that is, the decomposition is not
in general unique. Thus when establishing the relationship between
$\Delta$ and $\Theta$, a given state $s$ in $\Delta$ may play a
number of different roles.

From Definition~\ref{d:lifting}, the next two properties follows. In
fact, they are sometimes used in the literature as definitions of
lifting relations instead of being properties (see e.g.
\cite{SL94,LS91}).
\begin{theorem}\label{t:lifting.alternative}
\begin{enumerate}
\item Let $\Delta$ and $\Theta$ be distributions over $S$ and $T$,
  respectively. Then
  $\Delta \lift{\aRel} \Theta$ if and only if there
  exists a weight function $w:S\times T \rightarrow [0,1]$ such that
  \begin{enumerate}
  \item $\forall s\in S: \sum_{t\in T}w(s,t) = \Delta(s)$
  \item $\forall t\in T: \sum_{s\in S}w(s,t) = \Theta(t)$
  \item $\forall (s,t)\in S\times T: w(s,t) > 0 \Rightarrow s\aRel t$.
  \end{enumerate}
\item Let $\Delta,\Theta$ be distributions over  $S$ and
 $\aRel$ is an equivalence relation. Then
$\Delta \lift{\aRel} \Theta$ if and only if $\Delta(C)=\Theta(C)$
for all equivalence class $C\in S/{\cal R}$, where $\Delta(C)$
stands for the accumulation probability $\sum_{s\in C}\Delta(s)$.
\end{enumerate}
\end{theorem}
\begin{proof}
\begin{enumerate}
\item
($\Rightarrow$) Suppose $\Delta\lift{\aRel}\Theta$. By
Proposition~\ref{p:lifting}, we can decompose $\Delta$ and $\Theta$
such that $\Delta=\sum_{i\in I}p_i\cdot\pdist{s_i}$,
$\Theta=\sum_{i\in I}p_i\cdot\pdist{t_i}$, and $s_i\aRel t_i$ for
all $i\in I$. We define the weight function $w$ by letting
$w(s,t)=\sum\sset{p_i\mid s_i=s, t_i=t, i\in I}$ for any $s\in S,
t\in T$. This weight function can be checked to meet our
requirements.
\begin{enumerate}
\item For any $s\in S$, it holds that
\[\begin{array}{rcl}
\sum_{t\in T}w(s,t) &= & \sum_{t\in T}\sum\sset{p_i\mid s_i=s,
t_i=t, i\in I}\\
& = &\sum\sset{p_i\mid s_i=s, i\in I}\\
& = & \Delta(s)
\end{array}\]
\item Similarly, we have $\sum_{s\in S}w(s,t)=\Theta(t)$.
\item For any $s\in S, t\in T$, if $w(s,t)>0$ then there is some $i\in
I$ such that $p_i>0$, $s_i=s$, and $t_i=t$. It follows from
$s_i\aRel t_i$ that $s\aRel t$.
\end{enumerate}

($\Leftarrow$) Suppose there is a weight function $w$ satisfying the
three conditions in the hypothesis. We construct the index set
$I=\sset{(s,t)\mid w(s,t)>0, s\in S,t\in T}$ and probabilities
$p_{(s,t)}=w(s,t)$ for each $(s,t)\in I$.
\begin{enumerate}
\item It holds that $\Delta=\sum_{(s,t)\in I}p_{(s,t)}\cdot\pdist{s}$
because, for any $s\in S$,
\[\begin{array}{rcl}
(\sum_{(s,t)\in I}p_{(s,t)}\cdot\pdist{s})(s) & = & \sum_{(s,t)\in
I}w(s,t)\\
& = & \sum\sset{w(s,t)\mid w(s,t)>0, t\in T}\\
& = & \sum\sset{w(s,t)\mid t\in T}\\
& = & \Delta(s)
\end{array}\]
\item Similarly, we have $\Theta=\sum_{(s,t)\in
I}w(s,t)\cdot\pdist{t}$.
\item For each $(s,t)\in I$, we have $w(s,t)>0$, which implies $s\aRel
t$.
\end{enumerate}
Hence, the above decompositions of $\Delta$ and $\Theta$ meet the
requirement of the lifting $\Delta\lift{\aRel}\Theta$.

\item
($\Rightarrow$) Suppose $\Delta\lift{\aRel}\Theta$. By
Proposition~\ref{p:lifting}, we can decompose $\Delta$ and $\Theta$
such that $\Delta=\sum_{i\in I}p_i\cdot\pdist{s_i}$,
$\Theta=\sum_{i\in I}p_i\cdot\pdist{t_i}$, and $s_i\aRel t_i$ for
all $i\in I$. For any equivalence class $C\in S/{\cal R}$, we have
that
\[\begin{array}{rcl}
\Delta(C) = \sum_{s\in C}\Delta(s) & = & \sum_{s\in
C}\sum\sset{p_i\mid i\in I, s_i=s}\\
& = & \sum\sset{p_i\mid i\in I, s_i\in C}\\
& = & \sum\sset{p_i\mid i\in I, t_i\in C}\\
& = & \Theta(C)
\end{array}\]
where the equality in the third line is justified by the fact that
$s_i\in C$ iff $t_i\in C$ since $s_i\aRel t_i$ and $C\in S/{\cal
R}$.

($\Leftarrow$) Suppose, for each equivalence class $C\in S/{\cal
R}$, it holds that $\Delta(C)=\Theta(C)$. We construct the index set
$I=\sset{(s,t)\mid s\aRel t \mbox{ and } s,t\in S}$ and
probabilities
$p_{(s,t)}=\frac{\Delta(s)\Theta(t)}{\Delta([s]_{\aRel})}$ for each
$(s,t)\in I$, where $[s]_{\aRel}$ stands for the equivalence class
that contains $s$.
\begin{enumerate}
\item It holds that $\Delta=\sum_{(s,t)\in I}p_{(s,t)}\cdot\pdist{s}$
because, for any $s'\in S$,
\[\begin{array}{rcl}
(\sum_{(s,t)\in I}p_{(s,t)}\cdot\pdist{s})(s') & = & \sum_{(s',t)\in
I}p_{(s',t)}\\
& = & \sum\sset{\frac{\Delta(s')\Theta(t)}{\Delta([s']_{\aRel})}\mid s'\aRel t,\ t\in S}\\
& = & \sum\sset{\frac{\Delta(s')\Theta(t)}{\Delta([s']_{\aRel})}\mid t\in [s']_{\aRel}}\\
& = & \frac{\Delta(s')}{\Delta([s']_{\aRel})}\sum\sset{\Theta(t)\mid
t\in [s']_{\aRel}}\\
& = & \frac{\Delta(s')}{\Delta([s']_{\aRel})}\Theta([s']_{\aRel})\\
& = & \frac{\Delta(s')}{\Delta([s']_{\aRel})}\Delta([s']_{\aRel})\\
& = & \Delta(s')
\end{array}\]
\item Similarly, we have $\Theta=\sum_{(s,t)\in
I}p_{(s,t)}\cdot\pdist{t}$.
\item For each $(s,t)\in I$, we have $s\aRel
t$.
\end{enumerate}
Hence, the above decompositions of $\Delta$ and $\Theta$ meet the
requirement of the lifting $\Delta\lift{\aRel}\Theta$.
\end{enumerate}
\end{proof}

\section{Justifying the lifting operation}\label{s:justify}
In our opinion, the lifting operation given in
Definition~\ref{d:lifting} is not only concise but also on the right
track. This is justified by its intrinsic connection with some
fundamental concepts in mathematics, notably the Kantorovich metric.
\subsection{Justification by the Kantorovich metric}\label{s:kan}

We begin with some historical notes. The \emph{transportation
problem} has been playing an important role in linear programming
due to its general formulation and methods of
solution. The original transportation problem, formulated by the French mathematician G. Monge in 1781 \cite{Mon1781}, 
consists of finding an optimal way of shovelling a pile of sand into
a hole of the same volume. 
In the 1940s, the Russian mathematician and economist L.V.
Kantorovich, who was awarded a Nobel prize in economics in 1975 for
the theory of optimal allocation of resources, gave a relaxed
formulation of the problem and proposed a variational principle for
solving the problem \cite{Kan42}. Unfortunately, Kantorovich's work
went unrecognized during a long period of time. The later known
\emph{Kantorovich metric} has appeared in the literature under
different names, because it has been rediscovered historically
several times from different perspectives. Many metrics known in
measure theory, ergodic theory, functional analysis, statistics,
etc. are special cases of the general definition of the Kantorovich
metric \cite{Ver06}. The elegance of the formulation, the
fundamental character of the optimality criterion, as well as the
wealth of applications, which keep arising, place the Kantorovich
metric in a prominent position among the mathematical works of the
20th century. In addition, this formulation can be computed in
polynomial time \cite{Orl88}, which is an appealing feature for its
use in solving applied problems. For example, it is widely used to
solve a variety of problems in business and economy such as market
distribution, plant location, scheduling problems etc. In recent
years the metric attracted the attention of computer scientists
\cite{DD09}: it has been used in various different areas in computer
science such as probabilistic concurrency, image retrieval, data
mining, bioinformatics, etc.

Roughly speaking,  the Kantorovich metric provides a way of
measuring the distance between two distributions. Of course, this
requires first a notion of distance between the basic elements that
are aggregated into the distributions, which is often referred to as
the \emph{ground distance}. In other words, the Kantorovich metric
defines a ``lifted" distance between two distributions of mass in a
space that is itself endowed with a ground distance. There are a
host of metrics available in the literature (see e.g. \cite{GS02})
to quantify the distance between probability measures; see
\cite{Rac91} for a comprehensive review of metrics in the space of
probability measures. The Kantorovich metric has an elegant
formulation and a natural interpretation in terms of the
transportation problem.

We now recall the mathematical definition of the Kantorovich metric.
Let $(X,m)$ be a separable metric space. (This condition will be
used by Theorem~\ref{t:KRduality} below.)
\begin{definition}\label{d:K} Given any two Borel probability measures $\Delta$ and $\Theta$ on $X$, the \emph{Kantorovich distance} between
$\Delta$ and $\Theta$ is defined by
\[K(\Delta,\Theta)={\rm sup}\left\{\left|\int f d\Delta - \int f d\Theta\right|: ||f||\leq 1\right\}.\]
where $||\cdot||$ is the \emph{Lipschitz semi-norm} defined by
$||f||={\rm sup}_{x\not=y}\frac{|f(x)-f(y)|}{m(x,y)}$ for a function
$f:X\rightarrow \Real$ with $\Real$ being the set of all real
numbers.
\end{definition}

The Kantorovich metric has an alternative characterisation. We
denote by $\BP(X)$ the set of all Borel probability measures on $X$
such that for all $z\in X$, if $\Delta\in\BP(X)$ then $\int_X
m(x,z)\Delta(x)<\infty$. We write $M(\Delta,\Theta)$ for the set of
all Borel probability measures on the product space $X\times X$ with
marginal measures $\Delta$ and $\Theta$, i.e. if $\Gamma\in
M(\Delta,\Theta)$ then $\int_{y\in X} d\Gamma(x,y)=d\Delta(x)$ and
$\int_{x\in X} d\Gamma(x,y)=d\Theta(y)$ hold.
\begin{definition}\label{d:L}
For $\Delta,\Theta\in\BP(X)$, we define the metric $L$ as follows:
\[L(\Delta,\Theta)={\rm inf}\left\{\int m(x,y)d \Gamma(x,y): \Gamma\in M(\Delta,\Theta)\right\}.\]
\end{definition}

\begin{lemma}
If $(X,m)$ is a separable metric space then $K$ and $L$ are metrics
on $\BP(X)$. \hfill\qed
\end{lemma}

The famous Kantorovich-Rubinstein duality theorem gives a dual
representation of $K$ in terms of  $L$.
\begin{theorem}
[Kantorovich-Rubinstein \cite{KR58}]\label{t:KRduality} If $(X,m)$
is a separable metric space then for any two distributions
$\Delta,\Theta\in\BP(X)$ we have
$K(\Delta,\Theta)=L(\Delta,\Theta)$. \hfill\qed
\end{theorem}

In view of the above theorem, many papers in the literature directly
take Definition~\ref{d:L} as the definition of the Kantorovich
metric. Here we keep the original definition, but it is helpful to
understand $K$ by using $L$. Intuitively, a probability measure
$\Gamma\in M(\Delta,\Theta)$ can be understood as a
\emph{transportation} from one unit mass distribution $\Delta$ to
another unit mass distribution $\Theta$. If the distance $m(x,y)$
represents the cost of moving one unit of mass from location $x$ to
location $y$ then the Kantorovich distance gives the optimal total
cost of transporting the mass of $\Delta$ to $\Theta$. We refer the
reader to \cite{Vil03} for an excellent exposition on the
Kantorovich metric and the duality theorem.

Many problems in computer science only involve finite state spaces,
so discrete distributions with finite supports are sometimes more
interesting than continuous distributions. For two discrete
distributions $\Delta$ and $\Theta$ with finite supports
$\{x_1,...,x_n\}$ and $\{y_1,...,y_l\}$, respectively, minimizing
the total cost of a discretised version of the transportation
problem reduces to the following linear programming problem:
\begin{equation}\label{e:lpr}
\begin{array}{ll}
\mbox{minimize} & \sum_{i=1}^n \sum_{j=1}^l \Gamma(x_i,y_j)m(x_i,y_j)\\
\mbox{subject to}\hspace{4mm} 
 &     \bullet\  \forall 1\leq i\leq n: \sum_{j=1}^l \Gamma(x_i,y_j)=\Delta(x_i)\\
 &     \bullet\  \forall 1\leq j\leq l: \sum_{i=1}^n \Gamma(x_i,y_j)=\Theta(y_j)\\
 &     \bullet\  \forall 1\leq i\leq n,1\leq j\leq l: \Gamma(x_i,y_j)\geq 0.
\end{array}
\end{equation}

Since (\ref{e:lpr}) is a special case of the discrete mass
transportation problem, some well-known polynomial time algorithm
like \cite{Orl88} can be employed to solve it, which is an
attractive feature for computer scientists.

\bigskip
Recall that a pseudometric
  is a function that yields a non-negative real number for each pair
  of elements and satisfies the following: $m(s,s)=0$,
  $m(s,t)=m(t,s)$, and $m(s,t)\leq m(s,u)+m(u,t)$, for any $s,t\in S$. We say a
  pseudometric $m$ is $1$-bounded if $m(s,t)\leq 1$ for any $s$ and
  $t$.
Let $\Delta$ and $\Theta$ be distributions over a finite set $S$ of
states. In \cite{BW01} a $1$-bounded pseudometric $m$ on $S$ is
lifted to be a $1$-bounded pseudometric $\hat{m}$ on $\dist{S}$ by
setting the distance $\hat{m}(\Delta,\Theta)$ to be the value of the
following linear programming problem:
\begin{equation}\label{e:pc}
\begin{array}{ll}
\mbox{maximize} & \sum_{s\in S} (\Delta(s)-\Theta(s))x_{s}\\
\mbox{subject to}\hspace{4mm} 
 &     \bullet\  \forall s,t\in S: x_s-x_t\leq m(s,t)\\
 &     \bullet\  \forall s\in S: 0\leq x_s\leq 1.
\end{array}
\end{equation}
This problem can be dualised and then simplified to yield the
following problem:
\begin{equation}\label{e:pc1}
\begin{array}{ll}
\mbox{minimize} & \sum_{s,t\in S}y_{st} m(s,t) \\
\mbox{subject to}\hspace{4mm} 
 &     \bullet\  \forall s\in S: \sum_{t\in S} y_{st}=\Delta(s)\\
 &     \bullet\  \forall t\in S: \sum_{s\in S}y_{st}= \Theta(t)\\
 &     \bullet\  \forall s,t\in S: y_{st}\geq 0.
\end{array}
\end{equation}
Now (\ref{e:pc1}) is in exactly the same form as (\ref{e:lpr}).

This way of lifting pseudometrics via the Kantorovich metric as
given in (\ref{e:pc1}) has an interesting connection with the
lifting of binary relations given in Definition~\ref{d:lifting}.

\begin{theorem}\label{t:metric.relation}
Let $R$ be a binary relation and $m$ a pseudometric on a state space
$S$ satisfying \begin{equation}\label{e:hy} s\ R\ t
\quad\mbox{iff}\quad m(s,t)=0 \end{equation} for any $s,t\in S$.
Then it holds that \[\Delta\ \lift{R}\ \Theta \quad\mbox{ iff }\quad
\hat{m}(\Delta,\Theta)=0\] for any distributions
$\Delta,\Theta\in\dist{S}$.
\end{theorem}
\begin{proof}
Suppose $\Delta\ \lift{R}\ \Theta$. From
Theorem~\ref{t:lifting.alternative}(1) we know there is a weight
function $w$ such that
  \begin{enumerate}
  \item $\forall s\in S: \sum_{t\in S}w(s,t) = \Delta(s)$
  \item $\forall t\in S: \sum_{s\in S}w(s,t) = \Theta(t)$
  \item $\forall s,t\in S: w(s,t) > 0 \Rightarrow s\ R\ t$.
  \end{enumerate}
By substituting $w(s,t)$ for $y_{s,t}$ in (\ref{e:pc1}), the three
constraints there can be satisfied. For any $s,t\in S$ we
distinguish two cases:
\begin{enumerate}
\item either $w(s,t)=0$
\item or $w(s,t)>0$. In this case we have $s\ R\ t$, which implies
$m(s,t)=0$ by (\ref{e:hy}).
\end{enumerate}
Therefore, we always have $w(s,t)m(s,t)=0$ for any $s,t\in S$.
Consequently, $\sum_{s,t\in S}w(s,t)m(s,t)=0$ and the optimal value
of the problem in (\ref{e:pc1}) must be $0$, i.e.
$\hat{m}(\Delta,\Theta)=0$, and the optimal solution is determined
by $w$.

The above reasoning can be reversed to show that the optimal
solution of (\ref{e:pc1}) determines a weight function, thus
$\hat{m}(\Delta,\Theta)=0$ implies $\Delta\ \lift{R}\ \Theta$.
\end{proof}

The above property will be used in Section~\ref{s:metric} to give a
metric characterisation of probabilistic bisimulation (cf.
Theorem~\ref{t:bm}).

\subsection{Justification by network flow}
The lifting operation discussed in Section~\ref{s:lift} is also
related to the maximum flow problem in optimisation theory. This was
already observed by Baier et al. in \cite{BEM00}.

 We briefly recall the basic definitions of
networks. More details can be found in e.g. \cite{Eve79}. A
\emph{network} is a tuple $\CN=(N,E,\bot,\top,c)$ where $(N,E)$ is a
finite directed graph (i.e. $N$ is a set of nodes and $E\subseteq
N\times N$ is a set of edges) with two special nodes $\bot$ (the
\emph{source}) and $\top$ (the \emph{sink}) and a \emph{capability}
$c$, i.e. a function that assigns to each edge $(v,w)\in E$ a
non-negative number $c(v,w)$. A \emph{flow function} $f$ for $\CN$
is a function that assigns to edge $e$ a real number $f(e)$ such
that
\begin{itemize}
\item $0\leq f(e)\leq c(e)$ for all edges $e$.
\item Let $\IN(v)$ be the set of incoming edges to node $v$ and
  $\OUT(v)$ the set of outgoing edges from node $v$. Then, for each
  node $v\in N\backslash\sset{\bot,\top}$,
  \[\sum_{e\in\IN(v)}f(e) ~= \sum_{e\in\OUT(v)}f(e).\]
\end{itemize}
The \emph{flow} $F(f)$ of $f$ is given by
\[F(f) ~= \sum_{e\in\OUT(\bot)}f(e) - \sum_{e\in\IN(\bot)}f(e).\]
The \emph{maximum flow} in $\CN$ is the supremum (maximum) over the
flows $F(f)$, where $f$ is a flow function in $\CN$.

 We will see that
the question whether $\Delta\lift{\aRel}\Theta$ can be reduced to a
maximum flow problem in a suitably chosen network. Suppose
$\aRel\subseteq S\times S$ and $\Delta,\Theta\in\dist{S}$. Let
$S'=\sset{s'\mid s\in
  S}$ where $s'$ are pairwise distinct new states, i.e. $s'\in S$ for
all $s\in S$. We create two states $\bot$ and $\top$ not contained
in $S\cup S'$ with $\bot\not=\top$. We associate with the pair
$(\Delta,\Theta)$
 the following network $\CN(\Delta,\Theta,\aRel)$. 
\begin{itemize}
\item The nodes are $N=S\cup S'\cup\sset{\bot,\top}$.
\item The edges are $E=\sset{(s,t')\mid
    (s,t)\in\aRel}\cup\sset{(\bot,s)\mid s\in S}\cup\sset{(s',\top)\mid
    s\in S}$.
\item The capability $c$ is defined by $c(\bot,s)=\Delta(s)$,
  $c(t',\top)=\Theta(t)$ and $c(s,t')=1$ for all $s,t\in S$.
\end{itemize}

\temporaryremove{ 
\begin{figure}
\psset{unit=0.8cm}
\begin{center}
\rput(-5.5,-1){\rnode{source}{\qdisk(0,0){2pt}}}
\rput(5.5,-1){\rnode{sink}{\qdisk(0,0){2pt}}}
\uput[l](-5.5,-1){$\bot$} \uput[r](5.5,-1){$\top$}
\rput(-2.5,1){\rnode{s1}{\qdisk(0,0){2pt}}} \uput[l](-2.5,1){$s_1$}
\rput(-2.5,-0.5){\rnode{s2}{\qdisk(0,0){2pt}}}
\uput[l](-2.5,-0.25){$s_2$} \rput(-2.5,-2){$\vdots$}
\rput(-2.5,-3.5){\rnode{sn}{\qdisk(0,0){2pt}}}
\uput[l](-2.5,-3.5){$s_n$}
\rput(2.5,1){\rnode{t1}{\qdisk(0,0){2pt}}} \uput[r](2.5,1){$s'_1$}
\rput(2.5,-0.5){\rnode{t2}{\qdisk(0,0){2pt}}}
\uput[r](2.5,-0.25){$s'_2$} \rput(2.5,-2){$\vdots$}
\uput[r](2.5,-3.5){$s'_n$}
\rput(2.5,-3.5){\rnode{tn}{\qdisk(0,0){2pt}}}
\ncline[nodesep=2mm]{->}{s1}{t1}\aput{:U}{$c_{11}$}
\ncline[nodesep=2mm]{->}{s1}{t2}\aput{:U}{$c_{12}$}
\ncline[nodesep=2mm]{->}{s1}{tn}\aput{:U}{$c_{1n}$}
\ncline[nodesep=2mm]{->}{s2}{t2}\aput{:U}{$c_{22}$}
\ncline[nodesep=2mm]{->}{s2}{tn}\bput{:U}{$c_{2n}$}
\ncline[nodesep=2mm]{->}{sn}{tn}\bput{:U}{$c_{nn}$}
\ncline[nodesep=2mm]{->}{source}{s1}\aput{:U}{$\Delta(s_1)$}
\ncline[nodesep=2mm]{->}{source}{s2}\bput{:U}{$\Delta(s_2)$}
\ncline[nodesep=2mm]{->}{source}{sn}\bput{:U}{$\Delta(s_n)$}
\ncline[nodesep=2mm]{->}{t1}{sink}\aput{:U}{$\Theta(s_1)$}
\ncline[nodesep=2mm]{->}{t2}{sink}\bput{:U}{$\Theta(s_2)$}
\ncline[nodesep=2mm]{->}{tn}{sink}\bput{:U}{$\Theta(s_n)$}
\uput[d](0,-4.5){$c_{ij}=1$ for all $i,j$}
\end{center}\vskip 4cm
\caption{The network $\CN(\Delta,\Theta,\aRel)$}\label{f:network}
\end{figure}
}

The next lemma appeared as Lemma 5.1 in \cite{BEM00}.
\begin{lemma}\label{l:lift.flow}
Let $S$ be a finite set, $\Delta,\Theta\in\dist{S}$ and
$\aRel\subseteq S\times S$. The following statements are equivalent.
\begin{enumerate}
\item There exists a weight function $w$ for $(\Delta,\Theta)$ with
  respect to $\aRel$.
\item The maximum flow in $\CN(\Delta,\Theta,\aRel)$ is $1$.
\hfill\qed
\end{enumerate}
\end{lemma}

Since the lifting operation given in Definition~\ref{d:lifting} can
also be stated in terms of weight functions, we obtain the following
characterisation using network flow.
\begin{theorem}\label{t:lift.flow}
Let $S$ be a finite set, $\Delta,\Theta\in\dist{S}$ and
$\aRel\subseteq S\times S$. Then $\Delta\lift{\aRel}\Theta$ if and
only if the maximum flow in
  $\CN(\Delta,\Theta,\aRel)$ is $1$.
\end{theorem}
\begin{proof}
Combining Theorem~\ref{t:lifting.alternative}(1) and
Lemma~\ref{l:lift.flow}.
\end{proof}

The above property will play an important role in
Section~\ref{s:algo} to give an ``on the fly" algorithm for checking
probabilistic bisimilarity.

\section{Probabilistic bisimulation}\label{s:pbisi}
With a solid base of the lifting operation, we can proceed to define
a probabilistic version of bisimulation. We start with a
probabilistic generalisation of labelled transition systems (LTSs).
\begin{definition}
A \emph{probabilistic labelled transition
system}\index{probabilistic labelled transition system}
(pLTS)\footnote{Essentially the same model has appeared in the
literature under different names such as \emph{NP-systems}
\cite{JHW94}, \emph{probabilistic processes} \cite{JW95},
\emph{simple probabilistic automata} \cite{Seg95},
\emph{probabilistic transition systems} \cite{JW02} etc.
Furthermore, there are strong structural similarities with
\emph{Markov
  Decision Processes} \cite{Put94,DGMZ07}.} is a triple\\
$\langle S, \Act, \rightarrow \rangle$, where
\begin{enumerate}
\item $S$ is a set of states;
\item $\Act$ is a set of actions;
\item $\rightarrow \;\;\subseteq\;\; S \times \Act \times \dist{S}$ is the transition relation.
\end{enumerate}
As with LTSs, we usually write $s \ar{a} \Delta$ in place of
$(s,a,\Delta) \in \;\rightarrow$.
A pLTS is \emph{finitely branching} if for each state $s\in S$ the
set $\sset{\langle \alpha,\Delta\rangle\mid s\ar{\alpha}\Delta,
\alpha\in \Act, \Delta\in\dist{S}}$ is finite; if moreover $S$ is
finite, then the pLTS is \emph{finitary}.
\end{definition}

In a pLTS, one step of transition leaves a single state but might
end up in a set of states; each of them can be reached with certain
probability. An LTS may be viewed as a degenerate pLTS, one in which
only point distributions are used.

Let $s$ and $t$ are two states in a pLTS, we say $t$ can simulate
the behaviour of $s$ if the latter can exhibit action $a$ and lead
to distribution $\Delta$ then the former can also perform $a$ and
lead to a distribution, say $\Theta$, which can mimic $\Delta$ in
successor states. We are interested in a relation between two
states, but it is expressed by invoking a relation between two
distributions. To formalise the mimicking of one distribution by the
other, we make use of the lifting operation investigated in
Section~\ref{s:lift}.

\begin{definition}\label{d:sbisi2}
A relation $\aRel\subseteq S \times S$ is a {\em probabilistic
  simulation} if $s\ \aRel\ t$ implies
\begin{itemize}
\item if $s\ar{a}\Delta$ then there exists some $\Theta$ such that
  $t\ar{a}\Theta$ and $\Delta \lift{\aRel} \Theta$.
\end{itemize}
If both $\aRel$ and $\aRel^{-1}$ are probabilistic simulations, then
$\aRel$ is a {\em
  probabilistic bisimulation}. The largest probabilistic
  bisimulation, denoted by $\BISI$, is called \emph{probabilistic
  bisimilarity}.
\end{definition}

As in the nonprobabilistic setting, probabilistic bisimilarity can
be approximated by a family of inductively defined relations.
\begin{definition}
Let $S$ be the state set of a pLTS. We define:
\begin{itemize}
\item $\BISI_0:=S\times S$
\item $s\BISI_{n+1}t$, for $n\geq 0$, if
\begin{enumerate}
\item whenever $s\ar{a}\Delta$, there exists some $\Theta$ such
that $t\ar{a}\Theta$ and $\Delta\lift{\BISI_n} \Theta$;
\item whenever $t\ar{a}\Theta$, there exists some $\Delta$ such
that $s\ar{a}\Delta$ and $\Delta\lift{\BISI_n} \Theta$.
\end{enumerate}
\item $\BISI_{\omega}:=\bigcap_{n\geq 0}\BISI_n$
\end{itemize}
\end{definition}
In general, $\BISI$ is a strictly finer relation than
$\BISI_\omega$. However, the two relations coincide when limited to
finitely branching pLTSs.
\begin{proposition}\label{p:app}
On finitely branching pLTSs, $\BISI_\omega$ coincides with $\BISI$.
\end{proposition}
\begin{proof}
It is trivial to show by induction that $s\BISI t$ implies $s\BISI_n
t$ for all $n\geq 0$, thus $s\BISI_\omega t$.

Now we show that $\BISI_\omega$ is a bisimulation. Suppose
$s\BISI_\omega t$ and $s\ar{a}\Delta$. We have to show that there is
some $\Theta$ with $t\ar{a}\Theta$ and $\Delta\lift{\BISI_\omega}
\Theta$. Consider the set
\[T:=\{\Theta\mid t\ar{a}\Theta \wedge \Delta\not\lift{\BISI_\omega} \Theta\}.\]
For each $\Theta\in T$, we have $\Delta\not\lift{\BISI_\omega}
\Theta$, which means that there is some $n_{\Theta}> 0$ with
$\Delta\not\lift{\BISI_{n_{\Theta}}} \Theta$. Since $t$ is finitely
branching, $T$ is a finite set. Let $N=max\{n_{\Theta}\mid \Theta\in
T\}$. It holds that $\Delta\not\lift{\BISI_N} \Theta$ for all
$\Theta\in T$, since by a straightforward induction on $m$ we can
show that $s\BISI_n t$ implies $s\BISI_m t$ for all $m,n\geq 0$ with
$n>m$. By the assumption $s\BISI_\omega t$ we know that
$s\BISI_{N+1}t$. It follows that there is some $\Theta$ with
$t\ar{a}\Theta$ and $\Delta\lift{\BISI_N} \Theta$, so $\Theta\not\in
T$ and hence $\Delta\lift{\BISI_\omega} \Theta$. By symmetry we also
have that if $t\ar{a}\Theta$ then there is some $\Delta$ with
$s\ar{a}\Delta$ and $\Delta\lift{\BISI_\omega} \Theta$.
\end{proof}
Proposition~\ref{p:app} has appeared in \cite{Bai98}; here we have
given a simpler proof.

\section{Logical characterisation}\label{s:logic}
Let $\CL$ be a logic. We use the notation $\CL(s)$ to stand for the
set of formulae that state $s$ satisfies. This induces an
equivalence relation on states: $s\ =^\CL\ t$ iff $ \CL(s)=\CL(t)$.
Thus, two states are equivalent when they satisfy exactly the same
set of formulae.

In this section we consider two kinds of logical characterisations
of probabilistic bisimilarity.
\begin{definition}
[Adequacy and expressivity]
\begin{enumerate}
\item $\CL$ is \emph{adequate} w.r.t. $\BISI$ if
for any states $s$ and $t$,
\[s=^\CL t~~\mbox{iff}~~ s\BISI t.\]

\item $\CL$ is \emph{expressive} w.r.t. $\BISI$ if for
each state $s$ there exists a \emph{characteristic formula}
$\phi_s\in\CL$ such that, for any states $s$ and $t$,
\[t\models\phi_s~~\mbox{iff}~~s\BISI t.\]
\end{enumerate}
\end{definition}
We will propose a probabilistic extension of the Hennessy-Milner
logic, showing its adequacy, and then a probabilistic extension of
the modal mu-calculus, showing its expressivity.
\subsection{An adequate logic}\label{s:adequate}
We extend the Hennessy-Milner logic by adding a probabilistic choice
modality to express the bebaviour of distributions.

\begin{definition}
The class $\CL$ of modal formulae over $\Act$, ranged over by
$\phi$, is defined by the following grammar:
\[\begin{array}{rcl}
\phi & := & \top\mid\phi_1\wedge\phi_2 \mid \diam{a}\psi \mid
\neg\phi\\
\psi &:=& \bigoplus_{i\in I}p_i\cdot\phi_i \end{array}\] We call
$\phi$ a \emph{state formula} and $\psi$ a \emph{distribution
formula}. Note that a distribution formula $\psi$ only appears as
the continuation of a diamond modality $\diam{a}\psi$. We sometimes
use the finite conjunction $\bigwedge_{i\in I}\phi_i$ as a syntactic
sugar.

The \emph{satisfaction relation}\index{satisfaction relation}
$\models \subseteq S\times\CL $ is defined by
\begin{itemize}
\item $s\models \top$ for all $s\in S$.
\item $s\models \phi_1\wedge\phi_2 $ if $s\models\phi_i$ for
  $i=1,2$.
\item $s\models \diam{a}\psi$ if for some $\Delta\in \dist{S}$,
  $s\ar{a}\Delta$ and $\Delta\models\psi$.
\item $s\models\neg\phi$ if it is not the case that
  $s\models\phi$.
\item $\Delta\models\bigoplus_{i\in I}p_i\cdot\phi_i$ if there are
  $\Delta_i\in\dist{S}$, for all $i\in I, t\in\support{\Delta_i}$, with
  $t\models\phi_i$, such that $\Delta=\sum_{i\in I}p_i\cdot\Delta_i$.
\end{itemize}
\end{definition}
With a slight abuse of notation, we write $\Delta\models\psi$ above
to mean that $\Delta$ satisfies the distribution formula $\psi$.
The introduction of distribution formula distinguishes $\CL$ from
other probabilistic modal logics e.g. \cite{{JWL01,PS07}}.

It turns out that $\CL$ is adequate w.r.t. probabilistic
bisimilarity.
\begin{theorem}[Adequacy]\label{p:modal.characterisation}
Let $s$ and $t$ be any two states in a finitely branching pLTS. Then
$s\BISI t$ if and only if $s =^\CL t$.
\end{theorem}
\begin{proof}
($\Rightarrow$)
 Suppose $s\BISI t$, we show that $s\models\phi \Leftrightarrow
t\models\phi$ by structural induction on $\phi$.
\begin{itemize}
\item Let $s\models\top$, we clearly have $t\models\top$.
\item Let $s \models\phi_1\wedge\phi_2$. Then $s\models\phi_i$
  for $i=1,2$. So by induction $t\models\phi_i$, and we have
  $t\models \phi_1\wedge\phi_2$. By symmetry we also have
  $t\models \phi_1\wedge\phi_2$ implies $s\models
  \phi_1\wedge\phi_2$.
\item Let $s\models\neg\phi$. So $s\not\models\phi$, and by induction
  we have $t\not\models\phi$. Thus $t\models\neg\phi$. By symmetry we
  also have $t\not\models\phi$ implies $s\not\models\phi$.
\item Let $s\models\diam{a}\bigoplus_{i\in I}p_i\cdot\phi_i$. Then $s\ar{a}\Delta$ and $\Delta\models\bigoplus_{i\in I}p_i\cdot\phi_i$
  for some $\Delta$. So
  $\Delta=\sum_{i\in i}p_i\cdot\Delta_i$ and for all $i\in I$ and
  $s'\in\support{\Delta_i}$ we have $s'\models\phi_i$. Since $s\BISI t$, there is some $\Theta$ with $t\ar{a}\Theta$
  and $\Delta\lift{\BISI} \Theta$. By
  Proposition~\ref{prop:lifting}  we have that $\Theta=\sum_{i\in
  I}p_i\cdot\Theta_i$ and $\Delta_i\lift{\BISI}\Theta_i$. It follows
  that for each $t'\in\support{\Theta_i}$ there is some
  $s'\in\support{\Delta_i}$ with $s'\BISI t'$.
  So by induction we have $t'\models\phi_i$ for all $t'\in\support{\Theta_i}$ with $i\in I$.
  Therefore, we have $\Theta\models\bigoplus_{i\in
  I}p_i\cdot\phi_i$. It follows that $t\models\diam{a}\bigoplus_{i\in
  I}p_i\cdot\phi_i$.
  By symmetry we also have
  $t\models\diam{a}\bigoplus_{i\in I}p_i\cdot\phi_i\Rightarrow s\models\diam{a}\bigoplus_{i\in I}p_i\cdot\phi_i$.
\end{itemize}

($\Leftarrow$) We show that the relation $=^\CL$ is a probabilistic
bisimulation. Suppose $s =^\CL t$ and $s\ar{a}\Delta$. We have to
show that there is some $\Theta$ with $t\ar{a}\Theta$ and $\Delta
\lift{(=^\CL)} \Theta$. Consider the set
\[T:=\{\Theta \mid t\ar{a}\Theta \wedge \Theta=\sum_{s'\in\support{\Delta}}\Delta(s')\cdot \Theta
_{s'}\wedge \exists s'\in\support{\Delta},\exists
t'\in\support{\Theta_{s'}}: s'\not=^\CL t'\}\]
 For each $\Theta\in T$, there must be some
 $s'_\Theta\in\support{\Delta}$ and
 $t'_\Theta\in\support{\Theta_{s'_\Theta}}$ such that (i) either there is a formula $\phi_{\Theta}$ with
$s'_\Theta\models\phi_{\Theta}$ but
$t'_\Theta\not\models\phi_{\Theta}$ (ii) or there is a formula
$\phi'_{\Theta}$ with $t'_\Theta\models\phi'_{\Theta}$ but
$s'_\Theta\not\models\phi'_{\Theta}$. In the latter case we set
$\phi_{\Theta}=\neg\phi'_{\Theta}$ and return back to the former
case. So for each $s'\in\support{\Delta}$ it holds that
$s'\models\bigwedge_{\sset{\Theta\in T\mid s'_\Theta =
s'}}\phi_\Theta$ and for each $\Theta\in T$ with $s'_\Theta=s'$
there is some $t'_{\Theta}\in\support{\Theta_{s'}}$ with
$t'_{\Theta}\not\models \bigwedge_{\sset{\Theta\in T\mid s'_\Theta =
s'}}\phi_\Theta$. Let
\[\phi:=\diam{a}\bigoplus_{s'\in\support{\Delta}}\Delta(s')\cdot\bigwedge_{\sset{\Theta\in T\mid s'_{\Theta}=s'}}\phi_{\Theta}.\]
It is clear that $s\models\phi$, hence $t\models\phi$ by $s=^\CL t$.
It follows that there must be a $\Theta^\ast$ with
$t\ar{a}\Theta^\ast$,
$\Theta^\ast=\sum_{s'\in\support{\Delta}}\Delta(s')\cdot\Theta^\ast_{s'}$
and for each $s'\in\support{\Delta},
t'\in\support{\Theta^\ast_{s'}}$ we have
$t'\models\bigwedge_{\sset{\Theta\in T\mid s'_\Theta =
s'}}\phi_\Theta$. This means that $\Theta^\ast\not\in T$ and hence
for each $s'\in\support{\Delta}, t'\in\support{\Theta^\ast_{s'}}$ we
have $s'=^\CL t'$. It follows that $\Delta\lift{(=^\CL)}
\Theta^\ast$. By symmetry all transitions of $t$ can be matched up
by transitions of $s$.
\end{proof}

\subsection{An expressive logic}
We now add the probabilistic choice modality introduced in
Section~\ref{s:adequate} to the modal mu-calculus, and show that the
resulting probabilistic mu-calculus is expressive w.r.t.
probabilistic bisimilarity.
\subsubsection{Probabilistic modal mu-calculus} Let $\Var$ be a
countable set of variables.
We define a  set $\CL_\mu$ of modal formulae in positive normal form
given by the following grammar:
\[\begin{array}{rcl}
\phi & := &  \top\mid \bot \mid\diam{a}\psi \mid \boxm{a}\psi \mid \phi_1\wedge\phi_2 \mid \phi_1\vee\phi_2 \mid X \mid \mu X.\phi \mid \nu X.\phi\\
\psi & := & \bigoplus_{i\in I}p_i\cdot\phi_i\end{array}\]  where
$a\in\Act$, $I$ is a finite index set and $\sum_{i\in I}p_i=1$. Here
we still write $\phi$ for a state formula and $\psi$ a distribution
formula. Sometimes we also use the finite conjunction
$\bigwedge_{i\in I}\phi_i$ and disjunction $\bigvee_{i\in I}\phi_i$.
As usual, we have $\bigwedge_{i\in\emptyset}\phi_i=\top$ and
$\bigvee_{i\in \emptyset}\phi_i=\bot$.

The two fixed point operators $\mu X$ and $\nu X$ bind the
respective variable $X$. We apply the usual terminology of free and
bound variables in a formula and write $\fv(\phi)$ for the set of
free variables in $\phi$.

We use {\em environments}, which binds free variables to sets of
distributions, in order to give semantics to formulae. We fix a
finitary pLTS and let $S$ be its state set. Let
\[\Env=\setof{\rho}{\rho:\Var\rightarrow\Pow{S}}\] be the set of
all environments and ranged over by $\rho$. For a set $V\subseteq S$
and a variable $X\in\Var$, we write $\rho[X\mapsto V]$ for the
environment that maps $X$ to $V$ and $Y$ to $\rho(Y)$ for all
$Y\not=X$.

The semantics of a formula $\phi$ can be given as the set of states
satisfying it. This entails a semantic functional $\Op{\ }:\CL_\mu
\rightarrow \Env \rightarrow \Pow{S}$ defined inductively in Figure
\ref{f:semantics}, where we also apply $\Op{\ }$ to distribution
formulae and $\Op{\psi}$ is interpreted as the set of distributions
that satisfy $\psi$. As the meaning of a closed formula $\phi$ does
not depend on the environment, we write $\Op{\phi}$ for
$\Op{\phi}_\rho$ where $\rho$ is an arbitrary environment.

\begin{figure}
\[\begin{array}{rcl}
\Op{\top}_\rho & = & S \\
\Op{\bot}_\rho & = & \emptyset \\
\Op{\phi_1\wedge\phi_2}_\rho & = & \Op{\phi_1}_\rho\cap \Op{\phi_2}_\rho \\
\Op{\phi_1\vee\phi_2}_\rho & = & \Op{\phi_1}_\rho\cup \Op{\phi_2}_\rho\\
\Op{\diam{a}\psi}_\rho & = & \setof{s\in S}{\exists
  \Delta:s \ar{a} \Delta\ \wedge\ \Delta\in\Op{\psi}_\rho}\\
\Op{\boxm{a}\psi}_\rho & = & \setof{s\in S}{\forall
  \Delta:s \ar{a} \Delta\ \Rightarrow\
  \Delta\in\Op{\psi}_\rho}\\
\Op{X}_\rho & = & \rho(X) \\
\Op{\mu X.\phi}_\rho & = & \bigcap\setof{V\subseteq
S}{\Op{\phi}_{\rho[X\mapsto V]} \subseteq
  V}\\
\Op{\nu X.\phi}_\rho & = & \bigcup\setof{V\subseteq
S}{\Op{\phi}_{\rho[X\mapsto V]} \supseteq
  V}\\
  \Op{\bigoplus_{i\in I}p_i\cdot\phi_i}_\rho & = &  \setof{\Delta\in
\dist{S}}{\Delta= \bigoplus_{i\in I}p_i\cdot\Delta_i
  \ \wedge\ \forall i\in I,\forall t\in\support{\Delta_i}: t\in\Op{\phi_i}_\rho}
\end{array}\]
\caption{Semantics of probabilistic modal
mu-calculus}\label{f:semantics}
\end{figure}

The semantics of probabilistic modal mu-calculus (pMu) is the same
as that of the modal mu-calculus \cite{Koz83} except for the
probabilistic choice modality which are satisfied by distributions.
The characterisation of {\em least fixed point
  formula} $\mu X.\phi$ and {\em greatest fixed point formula} $\nu
X.\phi$ follows from the well-known Knaster-Tarski fixed point
theorem \cite{Tar55}.

We shall consider (closed) {\em equation systems} of formulae of the
form
\[\begin{array}{rcl}
E: X_1 & = & \phi_1 \\
       & \vdots & \\
   X_n & = & \phi_n
\end{array}\]
where $X_1,...,X_n$ are mutually distinct variables and
$\phi_1,...,\phi_n$ are formulae having at most $X_1,...,X_n$ as
free variables. Here $E$ can be viewed as a function
$E:\Var\rightarrow\CL_\mu$ defined by $E(X_i)=\phi_i$ for
$i=1,...,n$ and $E(Y)=Y$ for other variables $Y\in\Var$.

An environment $\rho$ is a {\em solution} of an equation system $E$
if $\forall i:\rho(X_i) = \Op{\phi_i}_\rho$. The existence of
solutions for an equation system can be seen from the following
arguments. The set $\Env$, which includes all candidates for
solutions, together with the partial order $\leq$ defined by
\[\rho\leq\rho'\ \mbox{\rm  iff}\ \forall
X\in\Var:\rho(X)\subseteq\rho'(X)\] forms a complete lattice. The
{\em equation functional} $\CE:\Env\rightarrow \Env$ given in the
$\lambda$-calculus notation by \[\CE:=\lambda\rho.\lambda
X.\Op{E(X)}_\rho\] is monotonic. Thus, the Knaster-Tarski fixed
point theorem guarantees existence of solutions, and the largest
solution
\[\rho_E:= \bigsqcup \setof{\rho}{\rho\leq\CE(\rho)}\]

\subsubsection{Characteristic equation systems}\label{s:ces} As
studied in \cite{SI94}, the behaviour of a process can be
characterised by an equation system of modal formulae. Below we show
that this idea also applies in the probabilistic setting.

\begin{definition}\label{d:cess}
Given a finitary  pLTS, its {\em characteristic equation system}
consists of one equation for each state $s_1,...,s_n\in S$.
\[\begin{array}{rcl}
E: X_{s_1} & = & \phi_{s_1} \\
       & \vdots & \\
   X_{s_n} & = & \phi_{s_n}
\end{array}\]
where
\begin{equation}\label{e:cf}
\phi_s:=(\bigwedge_{s\ar{a}\Delta}\diam{a}X_{\Delta})\wedge
(\bigwedge_{a\in\Act}\boxm{a}\bigvee_{s\ar{a}\Delta} X_{\Delta} )
\end{equation}
with $X_\Delta := \bigoplus_{s\in\support{\Delta}}\Delta(s)\cdot
X_s$.
\end{definition}

\begin{theorem}\label{t:ces}
Suppose $E$ is a characteristic equation system. Then $s\BISI t$ if
and only if $t\in \rho_E(X_s)$.
\end{theorem}
\begin{proof}
($\Leftarrow$) Let $\aRel=\setof{(s,t)}{t\in \rho_E(X_s)}$. We
  first show that
\begin{equation}\label{e:left}
\Theta\in\Op{X_\Delta}_{\rho_E}\ {\rm implies}\
\Delta\lift{\aRel}\Theta.
\end{equation}
Let $\Delta=\bigoplus_{i\in I}p_i\cdot \pdist{s_i}$, then $X_\Delta
= \bigoplus_{i\in I}p_i\cdot X_{s_i}$. Suppose
$\Theta\in\Op{X_\Delta}_{\rho_E}$. We have that
$\Theta=\bigoplus_{i\in I}p_i\cdot\Theta_i$ and, for all $i\in I$
and $t'\in\support{\Theta_i}$, that $t'\in \Op{X_{s_i}}_{\rho_E}$,
i.e. $s_i \aRel t'$. It follows that $\pdist{s_i} \lift{\aRel}
\Theta_i$ and thus
 $\Delta \lift{\aRel} \Theta$.

Now we show that $\aRel$ is a bisimulation.
\begin{enumerate}
\item Suppose $s\aRel t$ and $s\ar{a}\Delta$. Then
  $t\in\rho_E(X_s)=\Op{\phi_s}_{\rho_E}$. It follows from
  (\ref{e:cf}) that $t\in\Op{\diam{a}X_\Delta}_{\rho_E}$. So
  there exists some $\Theta$ such that $t\ar{a}\Theta$ and
  $\Theta \in \Op{X_\Delta}_{\rho_E}$. Now we apply (\ref{e:left}).

\item Suppose $s\aRel t$ and $t\ar{a}\Theta$. Then
  $t\in\rho_E(X_s)=\Op{\phi_s}_{\rho_E}$. It follows from
  (\ref{e:cf}) that $t\in\Op{\boxm{a}\bigvee_{s\ar{a}\Delta}
  X_{\Delta}}$. Notice that it
  must be the case that $s$ can enable action $a$, otherwise, $t\in
  \Op{[a]\bot}_{\rho_E}$ and thus $t$ cannot enable $a$ either, in contradiction
  with the assumption $t\ar{a}\Theta$. Therefore,
  $\Theta\in\Op{\bigvee_{s\ar{a}\Delta}X_{\Delta}}_{\rho_E}$,
  which implies
  $\Theta\in\Op{X_\Delta}_{\rho_E}$ for some $\Delta$ with
  $s\ar{a}\Delta$.
  Now we apply (\ref{e:left}).
\end{enumerate}

($\Rightarrow$) We define the environment $\rho_\BISI$ by
 \[\rho_\BISI(X_s):=\setof{t}{s\BISI t}.\]
 It sufficies to show that $\rho_\BISI$ is a post-fixed point of
 $\CE$, i.e.
\begin{equation}\label{e:post}
\rho_\BISI \leq \CE(\rho_\BISI)
\end{equation}
because in that case we have $\rho_\BISI \leq \rho_E$, thus $s\BISI
t$  implies $t\in\rho_\BISI(X_s)$ which in turn implies $t\in
\rho_E(X_s)$.

We first show that
\begin{equation}\label{e:right}
\Delta\lift{\BISI}\Theta\ {\rm implies}\
\Theta\in\Op{X_\Delta}_{\rho_\BISI}.
\end{equation}
Suppose $\Delta\lift{\BISI}\Theta$, by Proposition~\ref{p:lifting}
we have that (i) $\Delta=\bigoplus_{i\in I}p_i\cdot\pdist{s_i}$,
(ii) $\Theta=\bigoplus_{i\in I}p_i\cdot\pdist{t_i}$, (iii) $s_i\BISI
t_i$ for all $i\in I$. We know from (iii) that
$t_i\in\Op{X_{s_i}}_{\rho_\BISI}$. Using (ii) we have that
$\Theta\in\Op{\bigoplus_{i\in I}p_i\cdot X_{s_i}}_{\rho_\BISI}$.
Using (i) we obtain $\Theta\in\Op{X_\Delta}_{\rho_\BISI}$.

Now we are in a position to show (\ref{e:post}). Suppose $t\in
\rho_\BISI(X_s)$. We must prove that $t\in
\Op{\phi_s}_{\rho_\BISI}$, i.e.
\[t\in
(\bigcap_{s\ar{a}\Delta}\Op{\diam{a}X_\Delta}_{\rho_\BISI}) \cap
(\bigcap_{a\in\Act}\Op{\boxm{a}\bigvee_{s\ar{a}\Delta}X_{\Delta}}_{\rho_\BISI})\]
by (\ref{e:cf}). This can be done by showing that $t$ belongs to
each of the two parts of this intersection.
\begin{enumerate}
\item In the first case, we assume that $s\ar{a}\Delta$. Since $s\BISI
  t$, there exists some $\Theta$ such that $t\ar{a}\Theta$ and
  $\Delta \lift{\BISI} \Theta$. By (\ref{e:right}), we get
  $\Theta\in\Op{X_\Delta}_{\rho_\BISI}$. It follows that $t\in
  \Op{\diam{a}X_\Delta}_{\rho_\BISI}$.

\item In the second case, we suppose $t\ar{a}\Theta$ for any action $a\in\Act$ and distribution $\Theta$. Then by $s\BISI t$ there exists some
 $\Delta$ such that   $s\ar{a}\Delta$ and $\Delta \lift{\BISI}
 \Theta$. By (\ref{e:right}), we get
 $\Theta\in\Op{X_\Delta}_{\rho_\BISI}$. As a consequence,
 $t\in \Op{\boxm{a}\bigvee_{s\ar{a}\Delta}X_{\Delta}}_{\rho_\BISI}$.
 Since this holds for arbitrary action $a$, our desired result
 follows.
\end{enumerate}
\end{proof}

\subsubsection{Characteristic formulae} So far we know how to
construct the characteristic equation system for a finitary pLTS. As
introduced in \cite{Mul98}, the three transformation rules in
Figure~\ref{f:rules} can be used to obtain from an equation system
$E$ a formula whose interpretation coincides with the interpretation
of $X_1$ in the greatest solution of $E$. The formula thus obtained
from a characteristic equation system is called a {\em
characteristic formula}.
\begin{theorem}
Given a characteristic equation system $E$, there is a
characteristic formula $\phi_s$ such that $\rho_E(X_s)=\Op{\phi_s}$
for any state $s$. \hfill\qed
\end{theorem}

The above theorem, together with the results in Section~\ref{s:ces},
gives rise to the following corollary.

\begin{corollary}
For each state $s$ in a finitary pLTS, there is a characteristic
formula $\phi_s $ such that $s\BISI t$ iff $t\in\Op{\phi_s }$.
\hfill\qed
\end{corollary}

\begin{figure}
\begin{enumerate}
\item Rule 1: $E \rightarrow F$
\item Rule 2: $E \rightarrow G$
\item Rule 3: $E \rightarrow H$ if $X_n\not\in\fv(\phi_1,...,\phi_n)$
\end{enumerate}

\[\begin{array}{rclrclrclrcl}
E: X_1 & = & \phi_1 \qquad & F: X_1 & = & \phi_1 \qquad & G: X_1 & =
&
       \phi_1[\phi_n/X_n] \qquad & H: X_1 & = & \phi_1 \\
       & \vdots & & & \vdots & &  & \vdots & & & \vdots & \\
   X_{n-1} & = & \phi_{n-1} & X_{n-1} & = & \phi_{n-1} & X_{n-1} & = &
       \phi_{n-1}[\phi_n/X_n] & X_{n-1} & = & \phi_{n-1} \\
   X_n & = & \phi_n & X_n & = & \nu X_n.\phi_n &  X_n & = & \phi_n & &
       &
\end{array}\]
\caption{Transformation rules}\label{f:rules}
\end{figure}

\section{Metric characterisation}\label{s:metric}
In the definition of probabilistic bisimulation probabilities are
treated as labels since they are matched only when they are
identical. One may argue that this does not provide a robust
relation: Processes that differ for a very small probability, for
instance, would be considered just as different as processes that
perform completely different actions. This is particularly relevant
to many applications where specifications can be given as perfect,
but impractical processes and other, practical processes are
considered acceptable if they only differ from the specification
with a negligible probability.

To find a more flexible way to differentiate processes, researchers
in this area have borrowed from mathematics the notion of
metric\footnote{For simplicity, in this section we use the term
metric to denote both metric and pseudometric.  All the results are
based on pseudometrics.}.  A metric is defined as a function that
associates a distance with a pair of elements. Whereas topologists
use metrics as a tool to study continuity and convergence, we will
use them to provide a measure of the difference between two
processes that are not quite bisimilar.

 Since different processes may behave  the same, they will be
  given distance zero in our metric semantics. So we are more
  interested in pseudometrics than metrics.

In the rest of this section, we fix a finite state pLTS
$(S,\Act,\ar{})$ and provide the set of pseudometrics on $S$ with
the following partial order.
\begin{definition}
The relation $\preceq$ for the set $\CM$ of $1$-bounded
pseudometrics on $S$ is defined by
\[m_1\preceq m_2\ {\rm if}\ \forall s,t: m_1(s,t)\geq
  m_2(s,t).\]
\end{definition}
Here we reverse the ordering with the purpose of
  characterizing bisimilarity as the {\em greatest} fixed point (cf:
  Corollary~\ref{c:bimx}).

\begin{lemma}\label{l:latt}
$(\CM,\preceq)$ is a complete lattice.
\end{lemma}
\begin{proof}
The top element is given by $\forall s,t:\top(s,t)=0$; the bottom
element is given by $\bot(s,t)=1$ if $s\not=t$, $0$ otherwise.
Greatest lower bounds are given by $(\bigsqcap
X)(s,t)=\sup\{m(s,t)\mid m\in X\}$ for any $X\subseteq\CM$. Finally,
least upper bounds are given by $\bigsqcup X=\bigsqcap\
\{m\in\CM\mid \forall m'\in X: m'\preceq m\}$.
\end{proof}

\begin{definition}\label{d:sm}
$m\in\CM$ is a {\em state-metric} if, for all
  $\epsilon\in [0,1)$, $m(s,t)\leq\epsilon$ implies:
\begin{itemize}
\item if $s\ar{a}\Delta$ then there exists some $\Delta'$ such that
  $t\ar{a}\Delta'$ and $\hat{m}(\Delta,\Delta')\leq\epsilon$
\end{itemize}
\end{definition}
where the lifted metric $\hat{m}$ was defined in (\ref{e:pc}) via
the Kantorovich metric. Note that if $m$ is a state-metric then it
is also a metric. By $m(s,t)\leq\epsilon$ we have
$m(t,s)\leq\epsilon$, which implies
\begin{itemize}
\item if $t\ar{a}\Delta'$ then there exists some $\Delta$ such that
  $s\ar{a}\Delta$ and $\hat{m}(\Delta',\Delta)\leq\epsilon$.
\end{itemize}
In the above definition, we prohibit $\epsilon$ to be $1$ because we
use $1$ to represent the distance between any two incomparable
states including the case where one state may perform a transition
and the other may not.

The greatest state-metric is defined as
\[m_{\it max}=\bigsqcup\{m\in\CM\mid m \mbox{ is a state-metric}\}.\]

It turns out that state-metrics correspond to bisimulations and the
greatest state-metric corresponds to bisimilarity. To make the
analogy closer, in what follows we will characterize $m_{\it max}$
as a fixed point of a suitable monotone function on $\CM$. First we
recall the definition of Hausdorff distance.
\begin{definition}
Given a $1$-bounded metric $d$ on $Z$, the {\em Hausdorff distance}
between two subsets $X,Y$ of $Z$ is defined as follows:
\[H_d(X,Y)=\max\{\sup_{x\in X}\inf_{y\in Y}d(x,y),\sup_{y\in Y}\inf_{x\in
  X}d(y,x)\}\]
where $\inf\ \emptyset=1$ and $\sup\ \emptyset=0$.
\end{definition}
Next we define a function $F$ on $\CM$ by using the Hausdorff
distance.

\begin{definition}
Let $der(s,a)=\{\Delta\mid s\ar{a}\Delta\}$. $F(m)$ is a
pseudometric given by:
\begin{displaymath}
F(m)(s,t)=\sup_{a\in\Act}\{H_{\hat{m}}(der(s,a),der(t,a))\}.
\end{displaymath}
\end{definition}

Thus we have the following property.

\begin{lemma}
For all $\epsilon\in [0,1)$, $F(m)(s,t)\leq\epsilon$ if and only if:
\begin{itemize}
\item if $s\ar{a}\Delta$ then there exists some $\Delta'$ such that
  $t\ar{a}\Delta'$ and $\hat{m}(\Delta,\Delta')\leq\epsilon$;
\item if $t\ar{a}\Delta'$ then there exists some $\Delta$ such that
  $s\ar{a}\Delta$ and $\hat{m}(\Delta',\Delta)\leq\epsilon$. \hfill\qed
\end{itemize}

\end{lemma}

The above lemma can be proved by directly checking the definition of
$F$, as can the next lemma.
\begin{lemma}\label{l:smfix}
$m$ is a state-metric if and only if $m\preceq F(m)$. \hfill\qed
\end{lemma}
Consequently we have the following characterisation:
\[m_{\it max}=\bigsqcup\{m\in\CM\mid m\preceq F(m)\}.\]

\begin{lemma}\label{l:mono}
$F$ is monotone on $\CM$. \hfill\qed
\end{lemma}

Because of Lemma \ref{l:latt} and \ref{l:mono}, we can apply
Knaster-Tarski fixed point theorem, which tells us that $m_{\it
max}$ is the greatest fixed point of $F$. Furthermore, by
Lemma~\ref{l:smfix} we know that $m_{\it max}$ is indeed a
state-metric, and it is the greatest state-metric.

We now show the correspondence between state-metrics and
bisimulations.


\begin{theorem}\label{t:bm}
Given a binary relation $\aRel$ and a pseudometric $m\in\CM$ on a
finite state pLTS such that
\begin{equation}\label{e:rm}
m(s,t)=\left\{\begin{array}{ll}
0 & \mbox{if $s\aRel t$}\\
1 & \mbox{otherwise.}
\end{array}\right.
\end{equation}
Then $\aRel$ is a probabilistic bisimulation if and only if $m$ is a
state-metric.
\end{theorem}
\begin{proof}
The result can be proved by using Theorem~\ref{t:metric.relation},
which in turn relies on Theorem~\ref{t:lifting.alternative} (1).
Below we give an alternative proof that uses
Theorem~\ref{t:lifting.alternative} (2) instead.

Given two distributions $\Delta,\Delta'$ over $S$, let us consider
how to compute $\hat{m}(\Delta,\Delta')$ if $\aRel$ is an
equivalence relation. Since $S$ is finite, we may assume that
$V_1,...,V_n\in S/{\cal R}$ are all the equivalence classes of $S$
under $\aRel$. If $s,t\in V_i$ for some $i\in 1..n$, then
$m(s,t)=0$, which implies $x_s=x_t$ by the first constraint of
(\ref{e:pc}). So for each $i\in 1..n$ there exists some $x_i$ such
that $x_i=x_s$ for all $s\in V_i$.
 Thus, some
summands of (\ref{e:pc}) can be grouped together and we have the
following linear program:
\begin{equation}\label{eq:b}
\sum_{i\in 1..n}(\Delta(V_i)-\Delta'(V_i))x_{i}
\end{equation}
with the constraint $x_{i}-x_{j}\leq 1$ for any $i,j\in 1..n$ with
$i\not=j$. Briefly speaking, if $\aRel$ is an equivalence relation
then $\hat{m}(\Delta,\Delta')$ is obtained by maximizing the linear
program (\ref{eq:b}).

($\Rightarrow$) Suppose $\aRel$ is a bisimulation and $m(s,t)=0$.
From the assumption in (\ref{e:rm}) we know that $\aRel$ is an
equivalence relation.
 By the definition of
$m$ we have $s\aRel t$. If $s\ar{a}\Delta$ then $t\ar{a}\Delta'$ for
some $\Delta'$ such that $\Delta\lift{\aRel}\Delta'$. To show that
$m$ is a state-metric it suffices to prove $m(\Delta,\Delta')=0$. We
know from $\Delta\lift{\aRel}\Delta'$  and
Theorem~\ref{t:lifting.alternative} (2) that
$\Delta(V_i)=\Delta'(V_i)$, for each $i\in 1..n$. It follows that
(\ref{eq:b}) is maximized to be $0$, thus
$\hat{m}(\Delta,\Delta')=0$.

($\Leftarrow$) Suppose $m$ is a state-metric and has the relation in
(\ref{e:rm}).
 Notice that $\aRel$ is
an equivalence relation. We show that it is a bisimulation. Suppose
$s\aRel t$, which means $m(s,t)=0$. If $s\ar{a}\Delta$ then
$t\ar{a}\Delta'$ for some $\Delta'$ such that
$\hat{m}(\Delta,\Delta')=0$. To ensure that
$\hat{m}(\Delta,\Delta')=0$, in (\ref{eq:b}) the following two
conditions must be satisfied.
\begin{enumerate}
\item No coefficient is positive. Otherwise, if
  $\Delta(V_i)-\Delta'(V_i)>0$ then (\ref{eq:b}) would be
  maximized to a value not less than
  $(\Delta(V_i)-\Delta'(V_i))$, which is greater than $0$.
\item It is not the case that at least one coefficient is negative and the
  other coefficients are either negative or $0$. Otherwise, by summing up
   all the coefficients, we would get
  \[\Delta(S)-\Delta'(S)<0\]
which contradicts the assumption that $\Delta$ and $\Delta'$ are
distributions over $S$.
\end{enumerate}

Therefore the only possibility is that all coefficients in
(\ref{eq:b}) are $0$, i.e., $\Delta(V_i)=\Delta'(V_i)$ for any
equivalence class $V_i\in S/\aRel$. It follows from
Theorem~\ref{t:lifting.alternative} (2) that
$\Delta\lift{\aRel}\Delta'$. So we have shown that $\aRel$ is indeed
a bisimulation.
\end{proof}

\begin{corollary}\label{c:bimx}
Let $s$ and $t$ be two states in a finite state pLTS. Then $s\BISI
t$ if and only if $m_{\it max}(s,t)=0$.
\end{corollary}
\begin{proof}
($\Rightarrow$) Since $\BISI$ is a bisimulation, by
Theorem~\ref{t:bm} there exists some state-metric $m$ such that
$s\BISI t$ iff $m(s,t)=0$. By the definition of $m_{\it max}$ we
have $m\preceq m_{\it max}$. Therefore $m_{\it max}(s,t)\leq
m(s,t)=0$.

($\Leftarrow$) From $m_{\it max}$ we construct a pseudometric $m$ as
follows.
\[m(s,t)=\left\{\begin{array}{ll}
                 0 \quad& \mbox{if $m_{\it max}(s,t)=0$}\\
                 1 & \mbox{otherwise.}
                \end{array}\right.\]
Since $m_{\it max}$ is a state-metric, it is easy to see that $m$ is
also a state-metric. Now we construct a binary relation $\aRel$ such
that $\forall s,s': s\aRel s'$ iff $m(s,s')=0$. If follows from
Theorem~\ref{t:bm} that $\aRel$ is a bisimulation. If $m_{\it
  max}(s,t)=0$, then $m(s,t)=0$ and thus $s\aRel t$. Therefore we have
the required result $s\BISI t$ because $\BISI$ is the largest
bisimulation.
\end{proof}

\section{Algorithmic characterisation}\label{s:algo}
In this section we propose an ``on the fly" algorithm for checking
if two states in a finitary pLTS are  bisimilar.

An important ingredient of the algorithm is to check if two
distributions are related by a lifted relation. Fortunately,
Theorem~\ref{t:lift.flow} already provides us a method for deciding
whether $\Delta\lift{\aRel}\Theta$, for two given distributions
$\Delta,\Theta$ and a relation $\aRel$. We construct the network
$\CN(\Delta,\Theta,\aRel)$ and compute the maximum flow with
well-known methods, as sketched in
Algorithm 1.

\begin{algorithm}
\caption{\textbf{Check}$(\Delta,\Theta,\aRel)$}
\begin{tabular}{l}
\emph{Input}: A nonempty finite set $S$, distributions\\
\qquad $\Delta,\Theta\in\dist{S}$
and $\aRel\subseteq S\times S$\\
\emph{Output}: If $\Delta\lift{\aRel}\Theta$ then ``yes'' else ``no''\\
\emph{Method}:\\
\qquad Construct the network $\CN(\Delta,\Theta,\aRel)$\\
\qquad Compute the maximum flow $F$ in $\CN(\Delta,\Theta,\aRel)$\\
\qquad If $F<1$ then return ``no'' else ``yes''.
\end{tabular}
\end{algorithm}
As shown in \cite{CHM90}, computing the maximum flow in a network
can be done in time $O(n^3/\log n)$ and space $O(n^2)$, where $n$ is
the number of nodes in the network. So we immediately have the
following result.
\begin{lemma}\label{l:check.lift}
The test whether $\Delta\lift{\aRel}\Theta$ can be done in time $
O(n^3/\log n)$ and space $O(n^2)$. \hfill\qed
\end{lemma}

\begin{algorithm}
\caption{\textbf{Bisim}$(s,t)$}
\parbox[c]{7.8cm}{
\begin{algorithmic}
\STATE \textbf{Bisim}$(s,t)=\{$ \STATE $NotBisim:=\sset{}$ \STATE
\textbf{fun} \textbf{Bis}$(s,t)$=\{ \STATE \qquad $Visited:=\sset{}$
\STATE \qquad $Assumed:=\sset{}$ \STATE \qquad
\textbf{Match}$(s,t)$\} \STATE\} \textbf{handle}
$WrongAssumption\Rightarrow\textbf{Bis}(s,t)$ \STATE \textbf{return}
\textbf{Bis}$(s,t)$

\bigskip
\STATE \textbf{Match}$(s,t)=$ \STATE
$Visited:=Visisted\cup\sset{(s,t)}$ \STATE $b=\bigwedge_{a\in
A}\textbf{MatchAction}(s,t,a)$ \IF{$b=false$} \STATE $NotBisim :=
NotBisim\cup\sset{(s,t)}$ \IF{$(s,t)\in Assumed$} \STATE
\textbf{raise} $WrongAssumption$ \ENDIF \ENDIF \STATE
\textbf{return} $b$

\bigskip
\STATE \textbf{MatchAction}$(s,t,a)=$ \FORALL{$s\ar{a}\Delta_i$}
\FORALL{$t\ar{a}\Theta_j$} \STATE
$b_{ij}=\textbf{MatchDistribution}(\Delta_i,\Theta_j)$ \ENDFOR
\ENDFOR \STATE \textbf{return} $(\bigwedge_i(\bigvee_j
b_{ij}))\underline{\wedge(\bigwedge_j(\bigvee_i b_{ij}))}$

\bigskip
\STATE\textbf{MatchDistribution}$(\Delta,\Theta)=$ \STATE Assume
$\support{\Delta}=\sset{s_1,...,s_n}$ and
$\support{\Theta}=\sset{t_1,...,t_m}$ \STATE
$\aRel:=\sset{(s_i,t_j)\mid \textbf{Close}(s_i,t_j)=true}$ \STATE
\textbf{return} \textbf{Check}$(\Delta,\Theta,\aRel)$

\bigskip
\STATE\textbf{Close}$(s,t)=$ \IF{$(s,t)\in NotBisim$} \STATE
\textbf{return} $false$ \ELSIF{$(s,t)\in Visited$} \STATE
$Assumed:=Assumed\cup\sset{(s,t)}$ \STATE \textbf{return} $true$
\ELSE \STATE \textbf{return} \textbf{Match}$(s,t)$ \ENDIF
\end{algorithmic}
}
\end{algorithm}

We now present a bisimilarity-checking algorithm by adapting the
algorithm proposed in \cite{Lin98} for value-passing processes,
which in turn was inspired by \cite{FM90}.

The main procedure in the algorithm is \textbf{Bisim}$(s,t)$. It
starts with the initial state pair $(s,t)$, trying to find the
smallest bisimulation relation containing the pair by matching
transitions from each pair of states it reaches.  It uses three
auxiliary data structures:
\begin{itemize}
\item $NotBisim$ collects all state pairs that have already been
  detected as not bisimilar.
\item $Visited$ collects all state pairs that have already been
  visited.
\item $Assumed$ collects all state pairs that have already been
  visited and assumed to be bisimilar.
\end{itemize}
The core procedure, \textbf{Match}, is called from function
\textbf{Bis} inside the main procedure \textbf{Bisim}. Whenever a
new pair of states is encountered it is inserted into $Visited$. If
two states fail to match each other's transitions then they are not
bisimilar and the pair is added to $NotBisim$. If the current state
pair has been visited before, we check whether it is in $NotBisim$.
If this is the case, we return $false$. Otherwise, a loop has been
detected and we make assumption that the two states are bisimilar,
by inserting the pair into $Assumed$, and return $true$. Later on,
if we find that the two states are not bisimilar after finishing
searching the loop, then the assumption is wrong, so we first add
the pair into $NotBisim$ and then raise the exception
$WrongAssumption$, which forces the function \textbf{Bis} to run
again, with the new information that the two states in this pair are
not bisimilar. In this case, the size of $NotBisim$ has been
increased by at least one. Hence, \textbf{Bis} can only be called
for finitely many times. Therefore, the procedure
\textbf{Bisim}$(s,t)$ will terminate. If it returns $true$, then the
set $(Visited - NotBisim)$ constitutes a bisimulation relation
containing the pair $(s,t)$.

The main difference from the algorithm of checking non-probabilistic
bisimilarity in \cite{Lin98} is the introduction of the procedure
\textbf{MatchDistribution}$(\Delta,\Theta)$, where we approximate
$\BISI$ by a  binary relation $\aRel$ which is coarser than $\BISI$
in general, and we check the validity of $\Delta\lift{\aRel}\Theta$.
If $\Delta\lift{\aRel}\Theta$ does not hold, then
$\Delta\lift{\BISI}\Theta$ is invalid either and
\textbf{MatchDistribution}$(\Delta,\Theta)$ returns \textit{false}
correctly. Otherwise, the two distributions $\Delta$ and $\Theta$
are considered equivalent with respect to $\aRel$ and we move on to
match other pairs of distributions. The correctness of the algorithm
is stated in the following theorem.

\begin{theorem}\label{t:correctness}
Given two states $s_0$ and $t_0$ in a finitary pLTS, the function
\textbf{Bisim}$(s_0,t_0)$ terminates, and it returns \emph{true} if
and only if $s_0\sim t_0$.
\end{theorem}
\begin{proof}
Let $\textbf{Bis}_i$ stand for the $i$-th execution of the function
\textbf{Bis}. Let $Assumed_i$ and $NotBisim_i$ be the set $Assumed$
and $NotBisim$ at the end of \textbf{Bis}$_i$. When \textbf{Bis}$_i$
is finished, either a $WrongAssumption$ is raised or no
$WrongAssumption$ is raised. In the former case, $Assumed_i\cap
NotBisim_i\not=\emptyset$; in the latter case, the execution of the
function \textbf{Bisim} is completed. From function \textbf{Close}
we know that $Assumed_i\cap NotBisim_{i-1}=\emptyset$. Now it
follows from the simple fact $NotBisim_{i-1}\subseteq NotBisim_i$
that $NotBisim_{i-1}\subset NotBisim_i$. Since we are considering
finitary pLTSs, there is some $j$ such that
$NotBisim_{j-1}=NotBisim_j$, when all the non-bisimilar state pairs
reachable from $s_0$ and $t_0$ have been found and \textbf{Bisim}
must terminate.

For the correctness of the algorithm, we consider the relation
$\aRel_i=Visited_i-NotBisim_i$, where $Visited_i$ is the set
$Visited$ at the end of \textbf{Bis}$_i$. Let \textbf{Bis}$_k$ be
the last execution of \textbf{Bis}. For each $i\leq k$, the relation
$\aRel_i$ can be regarded as an approximation of $\BISI$, as  far as
the states appeared in $\aRel_i$ are concerned. Moreover, $\aRel_i$
is a coarser approximation because if two states $s,t$ are
re-visited but their relation is unknown, they are assumed to be
bisimilar. Therefore, if \textbf{Bis}$_k(s_0,t_0)$ returns $false$,
then $s_0\not\BISI t_0$. On the other hand, if
\textbf{Bis}$_k(s_0,t_0)$ returns $true$, then $\aRel_k$ constitutes
a bisimulation relation containing the pair $(s_0,t_0)$. This
follows because $\textbf{Match}(s_0,t_0)=true$  which basically
means that whenever $s\aRel_k t$ and $s\ar{a}\Delta$ there exists
some transition $t\ar{a}\Theta$ such that
$\textbf{Check}(\Delta,\Theta,\aRel_k)=true$, i.e.
$\Delta\lift{\aRel_k}\Theta$. Indeed, this rules out the possibility
that $s_0\not\BISI t_0$ as otherwise we would have
$s_0\not\BISI_\omega t_0$ by Proposition~\ref{p:app}, that is
$s_0\not\BISI_n t_0$ for some $n>0$. The latter means that some
transition $s_0\ar{a}\Delta$ exists such that for all
$t_0\ar{a}\Theta$ we have $\Delta\not\lift{\BISI_{n-1}}\Theta$, or
symmetrically with the roles of $s_0$ and $t_0$ exchanged, i.e.
$\Delta$ and $\Theta$ can be distinguished at level $n$, so a
contradiction arises.
\end{proof}

Below we consider the time and space complexities of the algorithm.
\begin{theorem}\label{t:complexity}
Let $s$ and $t$ be two states in a pLTS with $n$ states in total.
The function $\textbf{Bisim}(s,t)$ terminates in time $O(n^7/\log
n)$ and space $O(n^2)$.
\end{theorem}
\begin{proof}
 The number
of state pairs is bounded by $n^2$. In the worst case, each
execution of the function $\textbf{Bis}(s,t)$ only yields one new
pair of states that are not bisimilar. The number of state pairs
examined in the first execution of $\textbf{Bis}(s,t)$ is at most
$O(n^2)$, in the second execution is at most $O(n^2-1)$, $\cdots$.
Therefore, the total number of state pairs examined is at most
$O(n^2+(n^2-1)+\cdots+1)=O(n^4)$. When a state pair $(s,t)$ is
examined, each transition of $s$ is compared with all transitions of
$t$ labelled with the same action. Since the pLTS is finitely
branching, we could assume that each state has at most $c$ outgoing
transitions. Therefore, for each state pair, the number of
comparisons of transitions is bound by $c^2$. As a comparison of two
transitions calls the function \textbf{Check} once, which requires
time $O(n^3/\log n)$ by Lemma~\ref{l:check.lift}. As a result,
examining each state pair takes time $O(c^2 n^3/\log n)$.
Finally, the worst case time complexity of executing
$\textbf{Bisim}(s,t)$ is $O(n^7/\log n)$.

The space requirement of the algorithm is easily seen to be
$O(n^2)$, in view of Lemma~\ref{l:check.lift}.
\end{proof}

\begin{remark} With mild modification, the above algorithm can be adapted to
check probabilistic similarity. We simply remove the underlined part
in the function $\textbf{MatchAction}$; the rest of the algorithm
remains unchanged. Similar to the analysis in
Theorems~\ref{t:correctness} and \ref{t:complexity}, the new
algorithm can be shown to correctly check probabilistic similarity
over finitary pLTSs; its worst case time and space complexities are
still $O(n^7/\log n)$ and  $O(n^2)$, respectively.
\end{remark}

\section{Conclusion}\label{s:conclude}
To define behavioural equivalences or preorders for probabilistic
processes often involves a lifting operation that turns a binary
relation $\aRel$ on states into a relation $\lift{\aRel}$ on
distributions over states. We have shown that several different
proposals for lifting relations can be reconciled. They are nothing
more than different forms of essentially the same lifting operation.
More interestingly, we have discovered that this lifting operation
corresponds well to the Kantorovich metric, a fundamental concept
used in mathematics to lift a metric on states to a metric on
distributions over states, besides the fact the lifting operation is
related to the maximum flow problem in optimisation theory.

The lifting operation leads to a neat notion of probabilistic
bisimulation, for which we have provided logical, metric, and
algorithmic characterisations.
\begin{enumerate}
\item We have
introduced a probabilistic choice modality to specify the behaviour
of distributions of states. Adding the new modality to the
Hennessy-Milner logic and the modal mu-calculus results in an
adequate and an expressive logic w.r.t. probabilistic bisimilarity,
respectively.

\item Due to the correspondence of the lifting operation and the
Kantorovich metric, bisimulations can be naturally characterised as
pseudometrics which are post-fixed points of a monotone function,
and bisimilarity as the greatest post-fixed point of the funciton.

\item We have presented an ``on the fly" algorithm to check if two
states in a finitary pLTS are bisimilar. The algorithm is based on
the close relationship between the lifting operation and the maximum
flow problem.
\end{enumerate}
In the belief that a good scientific concept is often elegant, even
seen from different perspectives, we consider the lifting operation
and probabilistic bisimulation as two concepts in probabilistic
concurrency theory that are formulated in the right way.

\bibliographystyle{abbrv}
\bibliography{main-bibfile}

\end{document}